\documentclass[reqno,11pt,a4paper]{amsart}

\usepackage{color}
\usepackage[numbers]{natbib}

\usepackage{hyperref}

\usepackage{amsmath, amsthm, amsfonts, amssymb}
\usepackage[all]{xy}
\usepackage{cancel}

\usepackage[applemac]{inputenc}

\theoremstyle{plain}
\newtheorem{theorem}{Theorem}
\newtheorem{proposition}[theorem]{Proposition}
\newtheorem{lemma}[theorem]{Lemma}

\theoremstyle{definition}
\newtheorem{example}{Example}

\theoremstyle{remark}

\renewcommand{\geq}{\geqslant}

\newcommand{\dev}{\partial}

\begin{document}

\title{Multidimensional Integrable Deformations of integrable PDEs}
\author[Casati]{Matteo Casati}
\email[Casati]{matteo@nbu.edu.cn}
\author[Zhang]{Danda Zhang}
\email[Zhang]{zhangdanda@nbu.edu.cn}
\address[Casati, Zhang]{School of Mathematics and Statistics\\Ningbo University\\Ningbo City, Zhejiang Province, PRC}
\begin{abstract}
In a recent series of papers by Lou et al., it was conjectured that higher dimensional integrable equations may be constructed by utilizing some conservation laws of $(1+1)$-dimensional systems. We prove that the deformation algorithm introduced in \cite{ljh22}, applied to Lax integrable $(1+1)$-dimensional systems, produces Lax integrable higher dimensional systems. The same property is enjoyed by the generalized deformation algorithm introduced in \cite{ljh23-1}; we present a novel example of a $(2+1)$-dimensional deformation of KdV equation obtained by generalized deformation. The deformed systems obtained by such procedure, however, pose a serious challenge because most of the mathematical structures that the $(1+1)$-dimensional systems possess is lost.
\end{abstract}
\maketitle

\tableofcontents

\section{Deformation algorithm}
A simple procedure to obtain, starting from integrable (1+1)-dimensional PDEs, higher dimensional integrable system has been introduced and applied to several examples in a recent series of papers \cite{wl23,ljh22,ljh23-1,ljh23-2}. As the authors remark, while the dimensional \emph{reduction} of a system is a relatively straightforward procedure, lifting a simple system to a more complicated one without losing its key properties requires much more caution.

The deformation algorithm first introduced in \cite{ljh22} proved itself reliable and versatile; the aim of this paper is presenting a proof that it indeed preserves the integrability of the system.

The original version of the algorithm applies to scalar evolutionary PDEs with conserved densities depending only on the solution (and not on its spatial derivatives); this can be generalized in a straightforward way to  vector-valued (i.e. multi-component systems of) PDEs.

Let us consider a $(1+1)$-dimensional evolutionary system of the form
\begin{equation}\label{eq:11}
u^i_t=F^i(u,u_x,\dots,u_{xn})
\end{equation}
with $u^i=u^i(t,x)$, $i=1,\ldots,N$ and $u_{xn}=\dev_x^nu$. We assume that the system possesses $(D-1)$ conservation laws of the form
\begin{equation}\label{eq:conslaw}
\dev_t\rho_\alpha(u)=\dev_x J_\alpha(u,u_x,\ldots,u_{xn}),\qquad\qquad\alpha=1,\ldots,D-1.
\end{equation}
In other words, $\rho_\alpha$ are conserved densities and $J_\alpha$ their corresponding fluxes or \emph{currents}  \cite{Ol93}; note that the densities $\rho_\alpha$ depend only on the variables $u^i$. 

Let us denote as $\dev_x$ the total derivative operator w.r.t. the spatial variable $x$, defined as usual on differential polynomials $f=f(u,u_x,\ldots,u_{xn})$
$$
\dev_x f=\sum_{i=1}^N\sum_{p\geq 0}\frac{\dev f}{\dev u^i_{xp}}u^i_{x(p+1)}.
$$
In the rest of the paper we will adopt the Einstein convention of skipping the summation symbol and regarding the sum over double indices as extended up to the ``natural'' number of terms.

The total time derivative $\dev_t$ acts, by the chain rule and according to \eqref{eq:11}, as
\begin{equation}
\dev_t f=\frac{\dev f}{\dev u^i_{xp}}\dev_t\dev_x^p u^i=\frac{\dev f}{\dev u^i_{xp}}\dev^p_x\dev_t u^i=\frac{\dev f}{\dev u^i_{xp}}\dev_x^pF^i.
\end{equation}
We introduce additional spatial variables $x^\alpha$, $\alpha=1,\ldots,D-1$ (with corresponding total derivatives $\dev_{x^\alpha}$: obviously $[\dev_x,\dev_{x^\alpha}]=[\dev_{x^\alpha},\dev_{x^\beta}]=0$) and consider functions $u^i$ of time $t$ and \emph{$D$ spatial variables} $(x,x^1,\ldots,x^{D-1})$.

Following the deformation algorithm \cite{ljh22}, we introduce the deformed operators
\begin{align*}
\hat{L}&=\dev_x+\rho_\alpha\dev_{x^\alpha},&\hat{T}&=\dev_t+\bar{J}_\alpha\dev_{x^\alpha}
\end{align*}
where we denote $\bar{J}_\alpha$ the deformed conserved currents $J_\alpha(u,\hat{L}u,\ldots,\hat{L}^nu)$. In general, we will denote (for both differential polynomials and operators)
\begin{equation}
\bar{f}:=f(\dev_x;u,\dev_x u,\ldots,\dev_x^n u)\big|_{\dev_x^n \to\hat{L}^n }.
\end{equation}

The deformed $(D+1)$ dimensional system is then given by
\begin{equation}\label{eq:def11}
\hat{T} u^i=\bar{F}^i
\end{equation}
or, in explicit evolutionary form,
\begin{equation}\label{eq:expldef}
u^i_t=\bar{F}^i-u^i_{x^\alpha}\bar{J}_{\alpha}.
\end{equation}
\subsection*{Main Theorem} Equation \eqref{eq:def11} is (Lax) integrable as a $(D+1)$-dimensional system if Equation \eqref{eq:11} is (Lax) integrable as a $(1+1)$-dimensional system and $\{u^i(t,x;x^1,\ldots,x^{D-1})\}_{i=1}^N$ is one of its solutions. 

\vspace{1em}
The proof of the Main Theorem is presented in the next Sections. The paper is organized as it follows: in Section \ref{sec:Lax} we show that \eqref{eq:def11} admits a Lax representation, obtained by a deformation of the original Lax pair of \eqref{eq:11}. The main technical result required to prove the statement is the commutativity of $\hat{L}$ and $\hat{T}$, which shows that they can replace everywhere the total $x$- and $t$-derivatives in the construction and in the proofs.
In Section \ref{sec:flows} we prove that the deformation algorithm preserves the integrable hierarchy of the original system; in practical terms, that all the higher symmetries of the system are deformed into symmetries. In Section \ref{sec:general} we briefly present the generalized deformation algorithm recently appeared in \cite{ljh23-1} and prove that the operators $\hat{L}$ and $\hat{T}$ commute in this case, too. As an example of a generalized deformation of an evolutionary integrable system, we present a novel $2$-dimensional deformation of the KdV equation. Finally, in the conclusion we offer our viewpoint on the multidimensional integrable equations obtained via the deformation algorithm.
\section{Lax integrability}\label{sec:Lax}
\begin{proposition}\label{thm:commTL}
Let $\{u^i(t,x;x^1,\ldots,x^{D-1})\}_{i=1}^N$ be a solution of \eqref{eq:11} and $\hat{L}$, $\hat{T}$ defined according to the deformation algorithm. Then
$[\hat{L},\hat{T}]=0$.
\end{proposition}
\begin{proof}
We recall that $\rho_\alpha$ and $J_\alpha$ are, respectively, conserved quantities and currents for Equation \eqref{eq:11}. This means that they satisfy Equation \eqref{eq:conslaw}. We now want to prove that the same relation holds true replacing the $x$-derivative with $\hat{L}$ and the $t$-derivative with $\hat{T}$, i.e. $\hat{T}\rho_\alpha=\hat{L}\bar{J}_\alpha$. By chain rule and Equation \eqref{eq:11}, we clearly have $\dev_t\rho_\alpha=(\dev_{u^i}\rho_\alpha) u^i_t=(\dev_{u^i}\rho_\alpha) F^i$. Similarly,
\begin{equation}\label{eq:Trho}
\hat{T}\rho_\alpha=\frac{\dev \rho_\alpha}{\dev u^i}\hat{T}u^i=\frac{\dev \rho_\alpha}{\dev u^i}\bar{F}^i=\left.\left(\dev_t\rho_\alpha\right)\right|_{\dev_x^n\to\hat{L}^n}.
\end{equation}
On the other hand, $\hat{L}\bar{J}_\alpha$ is of the same form of $\dev_xJ_\alpha$ after replacing all the $x$-derivatives with $\hat{L}$. We then have
\begin{equation}\label{eq:defocons}
\hat{T}\rho_\alpha-\hat{L}\bar{J}_\alpha=\left.\left(\dev_t\rho_\alpha-\dev_xJ_\alpha\right)\right|_{\dev^n_x\to\hat{L}^n}=0.
\end{equation}
Finally, an explicit computation for $[\hat{L},\hat{T}]$ yields
\begin{equation}
\begin{split}
[\hat{L},\hat{T}]&=\left(\left(\dev_x+\rho_\alpha\dev_{x^\alpha}\right)\bar{J}_\beta-\left(\dev_t+\bar{J}_\alpha\dev_{x^\alpha}\right)\rho_\beta\right)\dev_{x^\beta}\\
&=\left(\hat{L}\bar{J}_\beta-\hat{T}\rho_\beta\right)\dev_{x^\beta}.
\end{split}
\end{equation}
Therefore, the commutator vanishes due to \eqref{eq:defocons}.
\end{proof}

Let us consider the Lax pair for the original equation, arising as the compatibility condition for a linear system
\begin{equation}
\left\{\begin{array}{lcl} M\psi&=&0\\\psi_t&=&N\psi\end{array}\right.
\end{equation} 
with $M=M(\dev_x;u,u_x,\ldots,u_{xn})$, $N=N(\dev_x;u,u_x,\ldots,u_{xn})$. From $(\dev_t-N) M\psi=M(\dev_t-N) \psi$ we have the usual Lax pair representation
\begin{equation}
M_t=[N,M]
\end{equation}
which is equivalent to Equation \eqref{eq:11}. The existence of a Lax representation is a sufficient condition to characterize the integrability of the system.

Let us now introduce the deformed operator $\bar{M}$ and $\bar{N}$, built with the usual recipe, and consider the linear system
\begin{equation}\label{eq:defolinsyst}
\left\{\begin{array}{lcl} \bar{M}\psi&=&0\\\hat{T}\psi&=&\bar{N}\psi.\end{array}\right.
\end{equation}
\begin{proposition}
A compatibility condition for the system \eqref{eq:defolinsyst} is the Lax representation
\begin{equation}\label{eq:defolax}
\hat{T}(\bar{M})=[\bar{N},\bar{M}].
\end{equation}
\end{proposition}
\begin{proof}
From $\bar{M}\psi=0$ and $(\hat{T}-\bar{N})\psi=0$ we have the condition
$$
(\hat{T}-\bar{N})\bar{M}\psi-\bar{M}(\hat{T}-\bar{N})\psi=0.
$$
Now, note that $\bar{M}$ (resp. $\bar{N}$) is of the generic form $$\bar{M}=\sum_{k} a_k(u,\hat{L}u,\ldots)\hat{L}^k,$$ so $\hat{T}\circ\bar{M}$ is
$$
\hat{T}\left(\sum_{k} a_k \hat{L}^k\right)=\sum_{k} \hat{T}(a_k)\hat{L}^k+\sum_{k} a_k\hat{T}\hat{L}^k.
$$
From Proposition \ref{thm:commTL} we can rewrite the last term as $\bar{M}\hat{T}$. Denoting $\hat{T}(\bar{M})$ the operator where $\hat{T}$ acts only on the coefficients of $\bar{M}$, we can then rewrite the compatibility condition as
$$
\left(\hat{T}(\bar{M})-[\bar{N},\bar{M}]\right)\psi=0
$$
which is \eqref{eq:defolax} as claimed.
\end{proof}
The following Lemma extends the result we found in \textcolor{blue}{E}quation \eqref{eq:Trho} to any differential polynomial. Note, however, that the proof relies on Proposition \ref{thm:commTL}.
\begin{lemma}\label{thm:mainlemma}
Let $f=f(u,u_x,\dots,u_{xn})$. Then
$$
\hat{T}\bar{f}=\left(\dev_t f\right)\big|_{\dev_x^n u^i\to \hat{L}^nu^i}.
$$
\end{lemma}
\begin{proof}
By chain rule, recall that
$$
\dev_t f=\sum_{m\geq0}\frac{\dev f}{\dev u^i_{xm}}\dev_x^m u^i_t=\sum_{m\geq0}\frac{\dev f}{\dev u^i_{xm}}\dev_x^m F^i.
$$
We can use chain rule to compute $\hat{T}\bar{f}$, too. We have
\begin{equation}\label{eq:l4-1}
\hat{T}\bar{f}=\sum_{m,n_\alpha}\frac{\dev \bar{f}}{\dev\left(\dev_x^m\dev_{x^\alpha}^{n_\alpha}u^i\right)}\hat{T}(\dev_x^m\dev_{x^\alpha}^{n_\alpha}u^i).
\end{equation}
Moreover, we have
$$
\frac{\dev \bar{f}}{\dev\left(\dev_x^m\dev_{x^\alpha}^{n_\alpha}u^i\right)}=\frac{\dev \bar{f}}{\dev\left(\hat{L}^pu^j\right)}\frac{\dev\left(\hat{L}^pu^j\right)}{\dev\left(\dev_x^m\dev_{x^\alpha}^{n_\alpha}u^i\right)},
$$
where the sum is over $(p,j)$ and the first factor is $\overline{\dev f/\dev u^j_{xp}}$. On the other hand, by the same formula \eqref{eq:l4-1} we have
$$
\hat{T}\left(\hat{L}^p u^j\right)=\frac{\dev\left(\hat{L}^p u^j\right)}{\dev \left(\dev_x^m\dev_{x^\alpha}^{n_\alpha}u^i\right)}\hat{T}\left(\dev_x^m\dev_{x^\alpha}^{n_\alpha}u^i\right).
$$
We can then write
$$
\hat{T}\bar{f}=\left(\overline{\frac{\dev f}{\dev u^j_{xp}}}\right)\hat{T}(\hat{L}^pu^j).
$$
Finally, by Proposition \ref{thm:commTL} we rewrite the latter as
$$
\hat{T}\bar{f}=\left(\overline{\frac{\dev f}{\dev u^j_{xp}}}\right)\hat{L}^p\bar{F}^j=\left.\left(\frac{\dev f}{\dev u^j_{xp}}\dev_x^pF^j\right)\right|_{\dev_x^n\to\hat{L}^n}.
$$
\end{proof}
\begin{proposition}
The Lax representation \eqref{eq:defolax} is equivalent to the deformed equation $\hat{T}u^i=\bar{F}^i$.
\end{proposition}
\begin{proof}
By Lemma \ref{thm:mainlemma} we have 
$$
0=\hat{T}(\bar{M})-[\bar{N},\bar{M}]=\left(\dev_t M-[N,M]\right)|_{\dev_x^n\to\hat{L}^n}
$$
Now, the vanishing of the RHS is equivalent to $u^i_t=F^i$ with the replacement $\dev_x\to\hat{L}$, so to $\hat{T}u^i=\bar{F}^i$ as claimed.\end{proof}

\section{Integrable hierarchies and commuting flows}\label{sec:flows}
Proposition \ref{thm:commTL} can be interpreted as saying that the deformation algorithm corresponds to an implicit change of coordinates, after which $\hat{L}$ and $\hat{T}$ take the role of the (total) derivatives with respect to $x$ and $t$. Therefore, the following result is natural:
\begin{proposition}\label{thm:Tflow} Let $u^i_{t_n}=F^i_n$ and $u^i_{t_m}=F^i_m$ be two flows in the (one-dimensional) hierarchy associated to \eqref{eq:11}. For a conserved quantity $\rho_\alpha$, we denote $J_{n\alpha}$ (resp.~$J_{m\alpha}$) its corresponding currents for the $t_n$- (resp.~$t_m$-)flow. We then have
\begin{equation}\label{eq:commflows}
\hat{T}_n \bar{F}_m-\hat{T}_m\bar{F}_n=0,
\end{equation}
where we denote $\hat{T}_n=\dev_{t_n}+\bar{J}_{n\alpha}\dev_{x^\alpha}$ and $\hat{T}_m=\dev_{t_m}+\bar{J}_{m\alpha}\dev_{x^\alpha}$.
\end{proposition}
\begin{proof}
Lemma \ref{thm:mainlemma} holds true for both $\hat{T}_n$ and $\hat{T}_m$ (``$\hat{T}_n$- and $\hat{T}_m$-flow''), so we have $\hat{T}_n\bar{F}^i_m=\left.\left(\dev_{t_n}F^i_m\right)\right|_{\dev_x\to\hat{L}}$ and $\hat{T}_m\bar{F}^i_n=\left.\left(\dev_{t_m}F^i_n\right)\right|_{\dev_x\to\hat{L}}$. The commutativity of the original $t$-flows hence implies the commutativity of the $\hat{T}$-flows.
\end{proof}

Equation \eqref{eq:commflows} can be rewritten as $[\hat{T}_n,\hat{T}_m]u^i=0$. We plan to prove that the deformed $t$-flows, written as in Equation \eqref{eq:expldef}, are in involution as a $(D+1)$-dimensional evolutionary hierarchy, namely $[\dev_{t_n},\dev_{t_m}]u^i=0$. In order to do so, we first need the following Lemma:
\begin{lemma}\label{thm:lemmacurrent}
Let $J_{n\alpha}$ and $J_{m\alpha}$ as in Proposition \ref{thm:Tflow}, namely $\dev_{t_n}\rho_\alpha=\dev_x J_{n\alpha}$ and $\dev_{t_m}\rho_\alpha=\dev_x J_{m\alpha}$. Then
$$
\dev_{t_n}J_{m\alpha}-\dev_{t_m}J_{n\alpha}=0.
$$
\end{lemma}
\begin{proof}
Let take the $x-$derivative of the LHS. We have
$$
\dev_x\left(\dev_{t_n}J_{m\alpha}-\dev_{t_m}J_{n\alpha}\right)=\dev_{t_n}\dev_{t_m}\rho_\alpha-\dev_{t_m}\dev_{t_n}\rho_\alpha=0.
$$
Integrating the LHS of the previous equation from $-\infty$ to $x$ (which is the same that symbolically applying the primitive operator $\dev_x^{-1}$) under the standard hypothesis that $J_{m\alpha}$ vanish at $\pm\infty$, we obtain our claim.
\end{proof}
We can then prove the commutativity of the flows with respect to their time variable (recall that $\hat{T}$ contains also spatial derivatives); in particular we have $\dev_{t_m}=\hat{T}_m-\bar{J}_{m\alpha}\dev_{x^\alpha}$ (resp. for $t_n$).
\begin{proposition}
Let the $\hat{T}_m$ and $\hat{T}_n$ flows be in involution, namely $[\hat{T}_m,\hat{T}_n]u^i=0$. Then $[\dev_{t_m},\dev_{t_n}]u^i=0$.
\end{proposition}
\begin{proof}
We have
$$
0=[\hat{T}_m,\hat{T}_n]u^i=[\dev_{t_m}+\bar{J}_{m\alpha}\dev_{x^\alpha},\dev_{t_n}+\bar{J}_{n\beta}\dev_{x^\beta}]u^i.
$$
A direct computation (recall that $\dev_{t}$ represent evolutionary equations, so $[\dev_t,\dev_{x^\alpha}]=0$) gives
\begin{align*}
0=&\;\;[\dev_{t_m},\dev_{t_n}]u^i\\
&+\bar{J}_{n\alpha}[\dev_{t_m},\dev_{x^\alpha}]u^i-\bar{J}_{m\alpha}[\dev_{t_n},\dev_{x^\alpha}]u^i+\bar{J}_{m\alpha}\bar{J}_{n\beta}[\dev_{x^\alpha},\dev_{x^\beta}]u^i\\&+\left(\hat{T}_m\bar{J}_{n\alpha}-\hat{T}_n\bar{J}_{m\alpha}\right)u^i_{x^\alpha}.
\end{align*}
All the terms in the second line vanish for the aforementioned reason and the commutativity of spatial derivatives. The third line vanishes according to Lemmas \ref{thm:mainlemma} and \ref{thm:lemmacurrent}. Hence, the $t$-flows commute.
\end{proof}
\section{Generalized deformation algorithm}\label{sec:general}
In \cite{ljh23-1} the authors propose a generalized version of the deformation algorithm for non (necessarily)-evolutionary equations, like Camassa-Holm \cite{ch93}
$$
F(u,u_t,u_x,\ldots,u_{xn,tp})=0
$$
($u_{xn,tp}:=\dev_x^n\dev_t^pu$), possessing conserved quantities
\begin{align}\label{eq:gconsq}
\dev_t\rho_\alpha&=\dev_x J_\alpha\\\notag
\rho_\alpha&=\rho_\alpha(u,u_t,u_x,\ldots,u_{xn,tp})\\\notag
J_\alpha&=J_\alpha(u,u_t,u_x,\ldots,u_{xn,tp}).
\end{align}
Note in particular that we drop the requirement $\rho_\alpha=\rho_\alpha(u)$. The deformed ``$x$''- and ``$t$''-derivatives are
\begin{align*}
\hat{L}&=\dev_x+\bar{\rho}_\alpha\dev_{x^\alpha},&
\hat{T}&=\dev_t+\bar{J}_\alpha\dev_{x^\alpha},
\end{align*}
while the deformed equation is of the form
\begin{equation}
\bar{F}=F\left(u,\hat{T}u,\hat{L}u,\ldots,\hat{L}^n\hat{T}^p u\right)=0.
\end{equation}
As in Proposition \ref{thm:commTL} (which is a necessary condition for the validity of the algorithm), we have
$$
[\hat{L},\hat{T}]=0.
$$
A direct computation shows
\begin{equation}\label{eq:gcommTL}
[\hat{L},\hat{T}]=\left(\hat{L}\bar{J}_\alpha-\hat{T}\bar{\rho}_\alpha\right)\dev_{x^\alpha}.
\end{equation}
Similarly to the proof of Lemma \ref{thm:mainlemma}, we have
	
\begin{align*}
 \hat{L}\bar{J}_\alpha&=\frac{\dev\bar{J}_\alpha}{\dev\left(\hat{L}^p\hat{T}^q u\right)}\frac{\dev\left(\hat{L}^p\hat{T}^q u\right)}{\dev\left(\dev^n_x\dev^p_t\dev^{n_\alpha}_{x^\alpha}u\right)}\hat{L}\left(\dev^n_x\dev^p_t\dev^{n_\alpha}_{x^\alpha}u\right)\\
&=\frac{\dev\bar{J}_\alpha}{\dev\left(\hat{L}^m\hat{T}^p u\right)}\hat{L}\left(\hat{L}^m\hat{T}^p u\right)\\
&=\left.\left(\frac{\dev J_\alpha}{\dev u_{xm,tp}}\dev_x u_{xm,tp}\right)\right|_{\substack{{\dev_t\to\hat{T}}\\{\dev_x\to\hat{L}}}}=\left.\left(\dev_x J_\alpha\right)\right|_{\substack{{\dev_t\to\hat{T}}\\{\dev_x\to\hat{L}}}}.
\end{align*}

In the same way,
$$
\hat{T}\bar{\rho}_\alpha=\left.\left(\dev_t\rho_\alpha\right)\right|_{\substack{{\dev_t\to\hat{T}}\\{\dev_x\to\hat{L}}}}
$$
so that the RHS of \eqref{eq:gcommTL} vanishes because of \eqref{eq:gconsq}. These two results grant a straightforward generalization of all the results of Section \ref{sec:Lax} and Section \ref{sec:flows}, establishing the validity of the generalized deformation algorithm.
\begin{example} \textbf{A new deformation of KdV equation}. We consider the conserved quantity
		\begin{equation*}
			(\frac{1}{2}u_x^2-u^3)_t=(u_xu_{xxx}-\frac{1}{2}u_{xx}^2-3u^2u_{xx}+6uu_x^2-\frac{9}{2}u^4)_x,
		\end{equation*}
		which gives the deformed ``derivatives''
		\begin{align*}
			\hat{L}&=\dev_x+\bar{\rho}\dev_{y},&
			\hat{T}&=\dev_t+\bar{J}\dev_{y},
		\end{align*}
		where
		\begin{align}
			\bar{\rho}&=\frac{1}{2}	(\hat{L}u)^2-u^3,&
			\bar{J}&=(\hat{L}u)(\hat{L}^3u)-\frac{1}{2}(\hat{L}^2u)^2-3u^2\hat{L}^2u+6u(\hat{L}u)^2-\frac{9}{2}u^4.
		\end{align}
Note that the definition of $\hat{L}$ is implicit and that we need to solve 
\begin{equation}
\bar{\rho}=\frac12\left(u_x+\bar{\rho}u_y\right)^2-u^3.
\end{equation}
We obtain two possible solutions 
\begin{equation}
\bar{\rho}=\frac{-u_xu_y+1\pm \sqrt{2u^3u_y^2-2u_xu_y+1}}{u_y^2},
\end{equation}
but only the solution with the minus sign converges to $\rho$ for $u_y\to 0$; we then take it as our deformed conserved density, using which we define $\hat{L}$, $\bar{J}$ and $\hat{T}$.

Then we have the deformed equation
	\begin{equation}
u_t=\hat{L}(\hat{L}^2u+3u^2)-\bar{J}u_y,
\end{equation}
which possesses the Lax pair
	\begin{align*}
\hat{L}^2\psi+u\psi&=\lambda \psi,& \psi_t&=4\hat{L}^3\psi+6u\hat{L}\psi+3(\hat{L}u)\psi-\bar{J}\psi_y.
	\end{align*}

From the  KdV hierarchy
	\begin{align*}
	u_{t_{2n+1}}&=F_{2n+1}=\Phi^n u_x,&
	 \Phi&=\partial^2_x+4u+2u_x\partial^{-1}_x,
\end{align*}
we obtain the deformed  KdV hierarchy
$$
u_{t_{2m+1}}	=F_{2n+1}|_{\dev_x\to\hat{L}}-\bar{J}_{2n+1}u_y
$$
with
$\bar{J}_{2n+1}=J_{2n+1}|_{\dev_x\to\hat{L}}$
and
\begin{align*}
J_1&=\rho\\
J_3&=u_xu_{xxx}-\frac{1}{2}u_{xx}^2-3u^2u_{xx}+6uu_x^2-\frac{9}{2}u^4\\
J_5&=18u^5+30u^3u_{xx}-30u^2u_{xx}-14u_x^2u_{xx}+8uu_{xx}^2-16uu_xu_{xxx}\\
&\qquad-\frac12 u_{xxx}^3+3u^2u_{xxxx}+u_{xx}u_{xxxx}-u_xu_{xxxxx}\\
&\cdots\\
J_{2n+1}&=\dev_x^{-1}\left(u_x\dev_x-3u^2\right)\left(\dev_x^2+4u+2u_x\dev_x^{-1}\right)^nu_x\\
&\cdots
\end{align*}
the conserved currents of the one-dimensional KdV hierarchy.

The deformed KdV hierarchy enjoys the property
$$
[\dev_{t_{2m+1}},\dev_{t_{2n+1}}]u=0.
$$
This property can be explicitly checked,  but it should be noted that the complete form of the equation is very complicated. Even further reductions of the system, that can be obtained for instance assuming that $u$ does not depend on the $x$ variable, retain an extremely cumbersome form. 
\end{example}

\section{Discussion}\label{sec:discussion}
The deformation algorithm devised by Lou, Jia and Hao allows to produce $(D+1)$-dimensional integrable equations starting from the corresponding $(1+1)$-dimensional ones and $(D-1)$ conservation laws. In other words, the integrability of the higher-dimensional systems relies uniquely on the integrability of the one-dimensional ones, and the solutions $u(t,x,x^1,\ldots,x^{D-1})$ of the deformed system must be solutions, in particular, of the one-dimensional one. Indeed, the original equations and solutions can always be recovered from the deformed ones by dropping the dependency of $u$ from the spatial variables $x^\alpha$. Because of this reason, we choose to refer to these systems as \emph{integrable multidimensional deformations} of the (standard) integrable equations.

Moreover, it is important to remark that, while the deformation algorithm preserves the (Lax) integrability of the equation and therefore the existence of the higher symmetries, the (bi)Hamiltonian structure of the system is lost: it is immediate to observe that, while for instance $\dev_x$ is a Hamiltonian structure for many KdV-like systems, even in the scalar case $\hat{L}$ is not skewadjoint, and a fortiori not Hamiltonian. In general, the toolkit used to solve the traditional one-dimensional equations cannot be employed fully to address their higher dimensional deformations \cite{wl23}. However, the three dimensional deformation of the KdV equation presented in \cite{ljh22} admits at least nontrivial (namely, fully 3-dimensional) single soliton and travelling wave solutions. In a similar manner, a new type of peakon solution, as well as the travelling wave solution, has been obtained for the two dimensional deformation of Camassa-Holm equation in \cite{ljh23-1} and travelling wave and kink solutions of a two-dimensional deformation of Burgers equation have been obtained in \cite{wl23}.

As we pointed out, the deformation algorithm does not preserve the Hamiltonian structure of the equation; it does not even preserve the conserved quantities, in the following sense: it is of course true that, by \eqref{eq:defocons}, the $\hat{T}$-evolution of a conserved density is in the image of $\hat{L}$, but the expression does not vanish when integrated over all the spatial domain, because $\hat{L}$ itself is not a linear combination of total derivatives in $(x, x^\alpha)$. This is even more apparent when considering the $t$-evolution.  Indeed, from $\hat{T}\rho_\alpha=\hat{L}\bar{J}_\alpha$ we have
\begin{equation}\label{eq:consfail}
\dev_t\rho_\alpha=\dev_x\bar{J}_\alpha+\rho_\beta\dev_{x^\beta} \bar{J}_\alpha - \bar{J}_\beta\dev_{x^\beta}\rho_\alpha.
\end{equation} 
$\rho_\alpha$ is the density of a conserved quantity in the $t$-evolution if and only if $\dev_t\rho_\alpha$ is a total spatial divergence, i.e. it is of the form $\dev_x \tilde{J}_\alpha +\sum\dev_{x^\beta} f_{\alpha\beta}$ for some conserved $D$-dimensional current of components $(\tilde{J}_\alpha,f_{\alpha1},\dots,f_{\alpha(D-1)})$. This is not in general the case for \eqref{eq:consfail}, implying that a conserved quantity for the original system is no longer a conserved quantity of the deformed one. 

The discrepancy between the behaviour of the original $(1+1)$-dimensional integrable system and the deformed ones, notwithstanding the preservation of the integrability, offers the opportunity to explore these new hierarchies in search of their missing elements (i.e. their conserved quantities, possible Hamiltonian and bi-Hamiltonian structure, tau-structure, etc.)
\vskip1em
\paragraph{Acknowledgement}
This work was sponsored by the National Science Foundation of China (Grants no.~12101341, 11801289, 11435005, and 11975131) and K.~C.~Wong Magna Fund in Ningbo University. M.~C.~wishes to thank Profs.~J.~Ferapontov, Qu C.~and Dr V.~Novikov for their valuable discussions.


\begin{thebibliography}{99}
\bibitem{ch93} Camassa R and Holm D D. An integrable shallow water equation with peaked solitons. Phys. Rev. Lett. 71 1661 (1993)
\bibitem{Ol93} Olver P. Applications of Lie Groups to Differential Equations (Graduate Texts in Mathematics vol 107) New York:Springer (1993)
	\bibitem{wl23} Wang F R and Lou S Y. Lax integrable higher dimensional Burgers systems via a deformation algorithm and conservation laws. Chaos Solit. Fractals 169 (2023)
	\bibitem{ljh22} Lou S Y, Hao X Z, Man J. Deformation Conjecture: Deforming Lower Dimensional Integrable Systems to Higher Dimensional Ones by Using Conservation Laws. J. High Energy Phys 2023(03) (2023)	
	\bibitem{ljh23-1}
	Lou S Y, Jia M, and Hao X Z. Higher Dimensional Camassa-Holm Equations. Chinese Phys. Lett. 40, 020201 (2023)
\bibitem{ljh23-2}
	Lou S Y, Hao X Z and  Jia M. Higher Dimensional reciprocal Kaup-Newell systems. Acta Phys. Sin. 72(10), 100204 (2023)
		
		

	\end{thebibliography}
\end{document}